\documentclass[11pt]{article}

\usepackage[left=1.25in, right=1.25in, top=1in, bottom=1in, footskip=0.5in]{geometry}
\usepackage{amsmath, amssymb, amsfonts, amsthm, bm, wasysym}
\usepackage{enumerate}
\usepackage{graphicx, float, subcaption, pdflscape, afterpage}
\usepackage[labelfont={sf,bf}]{caption}
\usepackage{threeparttable, booktabs, placeins}
\usepackage[dvipsnames]{xcolor}
\usepackage{adjustbox, multirow, setspace}
\usepackage{xurl}
\usepackage[super]{nth}
\usepackage[utf8]{inputenc}
\usepackage[english]{babel}
\usepackage[colorlinks=true, linkcolor = BrickRed, urlcolor = blue, citecolor = PineGreen]{hyperref}
\usepackage[capitalize,noabbrev]{cleveref}
\usepackage[title]{appendix}
\allowdisplaybreaks

\theoremstyle{plain}
\newtheorem{theorem}{Theorem}[section]
\newtheorem{proposition}[theorem]{Proposition}
\newtheorem{lemma}[theorem]{Lemma}
\newtheorem{corollary}[theorem]{Corollary}
\theoremstyle{definition}

\newtheorem{assumption}{Assumption}[section]
\theoremstyle{remark}
\newtheorem{examplex}{Example}[section]
\newenvironment{example}
  {\pushQED{\qed}\examplex}
  {\popQED\endexamplex}
\newtheorem{remarkx}{Remark}[section]
\newenvironment{remark}
  {\pushQED{\qed}\remarkx}
  {\popQED\endexamplex}

\newcommand{\E}[2][]{\mathrm{E}_{#1}\left[ #2 \right]}
\newcommand{\En}[1]{\mathbb{E}_{n}\left[ #1 \right]}

\newcommand{\PP}[1]{\mathrm{P}_P\left( #1 \right)}
\renewcommand{\P}[1]{\mathrm{P} \left( #1 \right)}
\newcommand{\thh}{\hat{\theta}}
\newcommand{\tbb}{\Bar{\theta}}
\newcommand{\ehh}{\hat{\eta}}
\newcommand{\ghh}{\hat{\gamma}}
\newcommand{\ahh}{\hat{\alpha}}
\newcommand{\shh}{\hat{\sigma}}
\newcommand{\Gn}{\mathbb{G}_n}
\newcommand{\F}{\mathcal{F}}
\newcommand{\TT}{\mathcal{T}}
\newcommand{\R}{\mathcal{R}_n}

\newcommand{\RR}{\mathbb{R}}

\newcommand{\U}{\mathcal{U}}

\usepackage{natbib}
\title{\scshape Finite-Sample Guarantees \\ for High-Dimensional DML}
\author{Víctor Quintas-Martínez}
\date{\today}

\begin{document}
\renewcommand{\abstractname}{\vspace{-1em}}

\maketitle

\begin{abstract}
    \textsc{Abstract} \quad Debiased machine learning (DML) offers an attractive way to estimate treatment effects in observational settings, where identification of causal parameters requires a conditional independence or unconfoundedness assumption, since it allows to control flexibly for a potentially very large number of covariates. This paper gives novel finite-sample guarantees for joint inference on high-dimensional DML, bounding how far the finite-sample distribution of the estimator is from its asymptotic Gaussian approximation. These guarantees are useful to applied researchers, as they are informative about how far off the coverage of joint confidence bands can be from the nominal level. There are many settings where high-dimensional causal parameters may be of interest, such as the ATE of many treatment profiles, or the ATE of a treatment on many outcomes. We also cover infinite-dimensional parameters, such as impacts on the entire marginal distribution of potential outcomes. The finite-sample guarantees in this paper complement the existing results on consistency and asymptotic normality of DML estimators, which are either asymptotic or treat only the one-dimensional case.
\end{abstract}

\section{Introduction}
A recent strand of literature in econometrics has considered estimation of treatment effects and causal or structural parameters using machine learning (ML) methods \citep{chernozhukov2018automatic, belloni2017program,  athey2019machine, farrell2021deep}. In many observational settings, the treatment or policy whose impact we wish to quantify was not randomly assigned, so a simple comparison of treatment and control groups is confounded by factors that are correlated both with the outcome and with the treatment. We may still be able to identify the causal effect of the treatment, however, if we are willing to make an unconfoundedness assumption, i.e., that the treatment is exogenous conditional on an appropriate set of controls. There are several ways in which ML can be useful in such quasi-experimental research designs. On the one hand, it allows to control flexibly for a large number of covariates, as are typically available in modern datasets. On the other hand, it opens a wide range of possibilities in terms of which types of data can be used as controls, e.g., textual or image data.\footnote{This is increasingly relevant in applied economics. For example, see \citet{dube2020monopsony}, who use data from job ads in Amazon MTurk to study whether online employers have monopsonistic labor market power, controlling for job characteristics in the form of features learned by a random forest trained on the ad descriptions. Similarly, an IGC project by \citet{olken2017evidence}, uses Google Street View images to infer J-PAL survey responses in Indonesia; one could imagine using those same images as a proxy for socio-demographic controls.}

The first section of this paper considers inference on a set of parameters that can be expressed as averages of a functional,
\[\theta_{0j} = \E[P]{m_j(W, \gamma_{0j})}, \quad j = 1, \, \ldots, \, p, \]
where $\gamma_{0j}$ is an infinite-dimensional nuisance parameter (for example, a conditional expectation or regression function), and $p$ is potentially large. This setting encompasses, for example, joint inference on the effects of many different treatments or on the impact of a treatment on many different outcomes.

Modern ML methods perform very well in predictive settings, typically by trading off variance and bias through some form of explicit or implicit regularization (and so they offer an attractive methodology to estimate $\gamma_{0j}$ when it is some form of ``best predictor''). At the same time, that trade-off means that their convergence rates are slower than the parametric $\sqrt{n}$-rate, causing a first-order bias in the estimation of the target parameters $\theta_{0j}$ that does not, in general, vanish asymptotically. \citet{chernozhukov2018double} and subsequent work show how to construct estimators of $\theta_{0j}$ that are correctly centered asymptotically by ``debiasing'' the moment conditions above, making them first-order robust\footnote{More precisely, Neyman-orthogonal, as defined below.} to the ML estimation error. These are known as double or debiased machine learning (DML) estimators.

An applied research may be interested in conducting joint inference about a high-dimensional parameter $\{\theta_{0j}\}_{j=1, \, \ldots, \, p}$. That can be done by constructing simultaneous confidence bands that cover the true parameter with a pre-specified probability (approximately) $1-\alpha$. Based on DML estimators $\{\thh_{j}\}_{j = 1, \, \ldots, p}$ of $\{\theta_{0j}\}_{j = 1, \, \ldots, p}$, and an estimator $\shh_{j}^2$ of $\sigma_j^2 = \mathrm{Var}(\sqrt{n}(\thh_{j} - \theta_{0j}))$ for each $j = 1,\ldots,p$, we will consider a joint confidence band of the form $\times_{j=1}^p [\thh_j \mp n^{-1/2}\shh_jc_{\alpha}]$, where the critical value $c_{\alpha}$ is chosen so that the sup-$t$-statistic satisfies: \[\PP{\max_{1 \leq j \leq p} \sqrt{n}\frac{|\thh_{j} - \theta_{0j}|}{\shh_j} \leq c_\alpha} \approx 1 - \alpha.\] One way to choose $c_\alpha$ is to rely on the joint asymptotic normality of the $t$-statistics, which is a well-established fact in the literature (\citealp{chernozhukov2018double, chernozhukov2018automatic, chernozhukov2021automatic}), even in the high-dimensional case and for continua of parameters \citep{belloni2017program, belloni2018uniformly}.

The goal of this paper is to provide finite-sample guarantees for the normal approximation above, in the form of a bound on the Kolmogorov distance between the finite-sample distribution of the sup-$t$-statistic and $\max_{1 \leq j \leq p} Z_j$ for a suitable jointly Gaussian distribution $(Z_1, \ldots, Z_p) \sim \mathcal{N}(0, \Sigma)$, i.e.: \[\sup_{t \in \RR} \left| \PP{\max_{1 \leq j \leq p} \sqrt{n}\frac{(\thh_j - \theta_{0j})}{\sigma_j} \leq t} - \P{\max_{1 \leq j \leq p} Z_j \leq t}\right|.\] The dependence of the bound on $n$ and $p$ will be made explicit. These guarantees are useful to applied researchers, as they are informative about how far off the coverage of joint confidence bands can be from the nominal level for a given sample size and complexity of the problem. The closest existing result to ours is \citet{chernozhukov2021simple}, who provide finite-sample guarantees in the single-parameter case, $p = 1$, when using pure sample splitting.\footnote{A precise definition of what we mean by sample splitting will be provided in \cref{rem:entropy}.} We extend their results in two ways, allowing for a potentially large $p$ and for no sample splitting. To that end, we leverage results from the literature on high-dimensional estimation, including a maximal inequality for empirical processes \citep{chernozhukov2014gaussian} and new normal approximation results for high-dimensional vectors \citep{chernozhukov2021nearly}.

In the second section of this paper, we consider a more general moment problem,
\[\E[P]{\psi_u(W, \theta_{0u}, \eta_{0u})} = 0, \quad u \in \U,\] where $\U$ is potentially an uncountable set. In the continuum of target parameters case, the $\{\theta_{0u}\}_{u \in \U}$ could represent, for example, the marginal distribution of an outcome under a treatment, which allows to derive many other interesting statistics (e.g., quantile treatment effects, Gini coefficients, Oaxaca-Blinder decompositions of distributional shifts, etc.). Again, we wish to construct a simultaneous confidence band for $\{\theta_{0u}\}_{u \in \U}$ using DML, based on a normal approximation to the sup-$t$-statistic. We provide finite-sample guarantees for that normal approximation also in this setting. To the best of our knowledge, ours is the first paper to provide non-asymptotic guarantees for the DML estimator of a continuum of parameters, which extend and complement the asymptotic results of \citet{belloni2017program, belloni2018uniformly}.

\paragraph{Notation} Throughout the paper, we use the following notation. For a random variable $W \in \mathcal{W}$, distributed according given probability measure $P$ on $\mathcal{W}$, we denote by $\E[P]{\cdot}$ the expectation with respect to $P$, i.e., $\E[P]{f(W)} = \int f(w) \mathrm{d}P(w)$ for a suitably measurable and integrable $f$. We denote by $\En{\cdot}$ the average of a sample $\{W_i\}_{i=1}^n$ of size $n$, i.e., $\En{f(W)} = n^{-1}\sum_{i=1}^n f(W_i)$. We use $\Gn [\cdot] = \mathbb{G}_{n,p} [\cdot]$ for an empirical process $\sqrt{n}(\mathbb{E}_n[\cdot] - \mathrm{E}_P[\cdot])$ over a class $\mathcal{F}$ of suitable measurable and integrable functions $f : \mathcal{W} \rightarrow \RR$, i.e., \[\Gn[f] = \Gn[f(W)] = \frac{1}{\sqrt{n}}\sum_{i=1}^n \left(f(W_i) - \E[P]{f(W)} \right).\]
We denote by $L^q = L^q(P)$ the space of functions with finite $q$-th absolute moments with respect to $P$, $\{f: \mathcal{W} \rightarrow \RR \, : \, \E[P]{|f(W)|^q} < \infty\}$. For $f \in L^q(P)$, we denote by $\Vert f \Vert_{P,q} = (\E[P]{|f(W)|^q})^{1/q}$, the $L^q$ norm. For a bounded function $f : \mathcal{W} \rightarrow \RR$, we denote by $\Vert \cdot \Vert_{\infty}$ the sup norm, i.e., $\Vert f \Vert_{\infty} = \sup_{w \in \mathcal{W}} |f(w)|$.

For a function $f : \RR \rightarrow \RR$, we denote by $\partial_r f(r) = f'(r)$ the derivative with respect to $r$. For a functional $F : \F \rightarrow \RR$ over a class of functions $\F$, we define the Gateaux derivative of $F$ at $f$ in the direction $u \in \F$ as $\partial_r F(f + r u) |_{r = 0}$. We say that $F$ is Gateaux differential at $F$ if that derivative exists for all $u \in \F$.

For a function class $\F$ endowed with a norm $\Vert \cdot \Vert_{\F}$, and for any $\varepsilon > 0$, we define the covering number $N(\varepsilon, \F, \Vert \cdot \Vert_{\F})$ as the smallest number of closed balls with radius $\varepsilon$ that could cover $\F$. Denote by $F$ a measurable envelope for $\F$, i.e., a function such that $F \geq \sup_{f\in \F}|f|$. The uniform entropy number is, for any $\varepsilon > 0$,  $\log \sup_Q N(\varepsilon, \F, \Vert \cdot \Vert_{Q,2})$, where the supremum is taken over any finitely discrete probability measure $Q$ such that $\Vert F \Vert_{Q,2} > 0$.

\section{Averages of Many Linear Functionals}\label{sec:lin}
In this section we extend the results of \citet{chernozhukov2021simple} to the high-dimensional case. Suppose we have access to an i.i.d. sample $\{(Y_i, W_i)\}_{i=1}^n$ from a probability distribution $P$, where $Y \in \mathcal{Y}$ denotes outcomes of interest and $W \in \mathcal{W}$ are other observed data. Our goal is to construct a simultaneous confidence band for the set of (scalar) parameters $\{\theta_{0j}\}_{j=1, \, \ldots, \, p}$, satisfying the following moment condition:
\begin{equation}
    \theta_{0j} = \E[P]{m_j(W, \gamma_{0j})}, \quad j = 1, \, \ldots, \, p, \label{eq:moment}
\end{equation}
where the moment functional $m_j(\cdot, \cdot)$ may depend on the unknown true value of an infinite-dimensional nuisance parameter $\gamma_{0j} \in \Gamma$, e.g., a conditional expectation. The set $\Gamma \subset L^2$ is assumed to be a linear function space, and can be used to encode restrictions on $\gamma_{0j}$, such as smoothness \citep{chernozhukov2021simple}. We assume that an ML estimator $\ghh_j$ of $\gamma_{0j}$ can be obtained from the same data.

Below we give some examples where this setting may apply.

\begin{example}[ATE of many treatments]\label{ex:1}
Researchers may be interested in estimating the average treatment effect (ATE) of many different treatments, combinations of treatments or dosages. In observational settings, where treatments are not randomly assigned, ATEs may still be identified under an unconfoundedness assumption \citep{rosenbaum1983central}.

Let $D \in \mathcal{D}$ denote the treatment variable, where $\mathcal{D} = \{d_0, \ldots, d_p\}$ is the set of different treatment profiles, with $d_0$ denoting no treatment. Let $Y(d)$ for $d \in \mathcal{D}$ denote potential outcomes, so that the observed outcome is $Y = Y(D)$. Suppose that the researcher has access to a set of control variables $X$ such that $D$ is independent of $Y(d)$ given $X$ for all $d \in \mathcal{D}$. In that case, $\E[P]{Y(d)} = \E[P]{\gamma_0(d_j, X)}$, where $\gamma_0(d, x) = \E[P]{Y \mid D = d, X = x}$. (Notice that, with this formulation, $\gamma_0$ does not depend on $j$.) The ATE of treatment profile $d_j$ with respect to no treatment is: \[\theta_{0j} = \E[P]{m_j(W, \gamma_0)}, \quad m_j(w, \gamma_0) = \gamma_0(d_j, x) - \gamma_0(d_0, x), \quad j = 1, \, \ldots, p.\]
In modern datasets, the number of available controls $X$ may be large, so that ML methods may be especially well suited to estimate the nuisance parameter $\gamma_{0j}$.
\end{example}

\begin{example}[ATE on many outcomes]\label{ex:2}
Our framework also provides guarantees for uniform inference on the ATE of a binary treatment $D \in \{0,1\}$ on a large set of outcomes $Y = (Y_1, \ldots, Y_p)$ using the sup-$t$-statistic. In this case, under the appropriate unconfoundedness assumption, we have: \[\theta_{0j} = \E[P]{m_j(W, \gamma_{0j})}, \quad m_j(w, \gamma_{0j}) = \gamma_{0j}(1, x) - \gamma_{0j}(0, x), \quad j = 1, \, \ldots, p,\] where now $\gamma_{0j}(d,x) = \E[P]{Y_j \mid D = d, X = x}$.
\end{example}

\begin{example}[Policy optimization]\label{ex:3}
Consider again a binary treatment $D \in \{0,1\}$ and an outcome of interest $Y$. Policymakers may wish to select the best treatment assignment policy based on a set of characteristics $X$, where a treatment assignment policy is a mapping $\pi : X \rightarrow \{0,1\}$. Suppose we only have access to observational data, collected under an unknown treatment policy. We want to evaluate and compare the effects of an alternative set of candidate policies $\{\pi_1, \ldots, \pi_p\}$, where the average effect of policy $\pi_j$ is: \[\theta_{0j} = \E[P]{m_j(W,\gamma_0)}, \quad m_j(w, \gamma_{0j}) = \gamma_{0}(0, x) + \pi_j(X)(\gamma_{0}(1, x) - \gamma_{0}(0, x) ), \quad j = 1, \, \ldots, p.\] Here, $\gamma_0(d,x) = \E[P]{Y \mid D = d, X = x}$ again does not depend on $j$. Other recent literature has also considered doubly-robust approaches to policy optimization based on observational data (e.g., \citealp{athey2021policy}).
\end{example}

Given an estimator $\ghh_j$ of $\gamma_{0j}$, it may seem natural to estimate $\theta_{0j}$ by the empirical analog of \eqref{eq:moment}, \[\Check{\theta}_j = \En{m_j(W, \ghh_j)}, \quad j = 1, \, \ldots, \, p.\]
However, when $\ghh_j$ is obtained using modern ML methods, $\ghh_j$ typically converges to $\gamma_{0j}$ more slowly than the parametric $\sqrt{n}$-rate. In this formulation, \citet{chernozhukov2018double} show that the bias from using $\ghh_j$ instead of $\gamma_{0j}$ is of first-order magnitude, so that $\Check{\theta}_j$ is not asymptotically centered around the true value $\theta_{0j}$. Inference based on $\Check{\theta}_j$ will thus be incorrect if one fails to account for that.

Following \citet{chernozhukov2018double} and subsequent work, we proceed by adjusting the moment condition \eqref{eq:moment} to make it immune, to a first order, against the estimation error in $\ghh_j$. This method is known as double or debiased machine learning (DML). Below we collect some existing results (Propositions \ref{prop:rr} and \ref{prop:neyman}) that show how to ``debias'' the moment condition \eqref{eq:moment} in the particular case of linear, mean-square continuous functionals of $\gamma$.

\begin{assumption}[Linearity and mean-square continuity] \label{ass:linearity} For all $j = 1, \, \ldots, \, p$, the moment functional $\gamma \mapsto m_j(w, \gamma)$ is linear and mean-square continuous, i.e., there exists $\Bar{Q} < \infty$ such that \[\E[P]{m_j(W,\gamma)^2} \leq \Bar{Q}^2 \E[P]{\gamma(W)^2} \quad \text{for all } \gamma \in \Gamma, j = 1, \, \ldots, \, p.\]
\end{assumption}

\begin{proposition}[Riesz representation]\label{prop:rr} Suppose \cref{ass:linearity} holds. Then, there exists a unique $\alpha_{0j} \in \mathrm{cl}(\mathrm{span} (\Gamma))$ such that \[\E[P]{m_j(W,\gamma)} = \E[P]{\alpha_{0j}(W)\gamma(W)} \quad \text{for all } \gamma \in \Gamma, j = 1, \, \ldots, \, p.\]
\end{proposition}
\begin{proof}
This is a consequence of the Riesz representation theorem. For a proof under more general conditions in the context of classic semiparametric theory, see, e.g., \citet{newey1994asymptotic}, \citet{ichimura2022influence}. See also, e.g., \citet{chernozhukov2018automatic} for a more detailed discussion of the role of the Riesz representer in DML.
\end{proof}

\begin{remark}
It is easy to see that the functionals in examples \ref{ex:1} to \ref{ex:3} are linear. One can also show mean-square continuity under appropriate regularity conditions (e.g., an overlap condition, $0 < \underline{p} \leq \PP{D = d \mid X} \leq \overline{p} < 1$ a.s. for all possible treatments or treatment profiles $d$), see e.g., \citet{chernozhukov2018automatic}. As a consequence, the Riesz representer exists in examples \ref{ex:1} to \ref{ex:3}.
\end{remark}

In general, the true Riesz representer $\alpha_{0j}$ will be unknown. We assume that an estimator $\ahh_j$ can be obtained from the same data. In some cases, the explicit form of $\alpha_{0j}$ can be derived. For instance, it can be shown that \[\alpha^*_{0j}(D,X) = \frac{\mathbf{1}\{D = d_j\}}{\PP{D = d_j \mid X}} - \frac{\mathbf{1}\{D = d_0\}}{\PP{D = d_0 \mid X}}\] is a Riesz representer for the functional \cref{ex:1}.\footnote{With a restricted semiparametric model $\Gamma$, it is possible that $\alpha^*_{0j} \notin \mathrm{cl}(\mathrm{span} (\Gamma))$, in which case the \textit{unique} or minimal Riesz representer of \cref{prop:rr} would be $\alpha_{0j} = \mathrm{Proj}(\alpha^*_{0j} \mid \mathrm{cl}(\mathrm{span} (\Gamma)))$.} An estimator $\ahh_j$ can be then constructed by plugging in a non-parametric estimate of the propensity scores $\PP{D = d \mid X}$. A more recent strand of literature considers automatic estimation of $\alpha_{0j}$, where knowledge of the explicit form of $\alpha_{0j}$ is not required \citep{chernozhukov2018automatic, chernozhukov2020adversarial, chernozhukov2021automatic}.

We consider a point estimator for $\{\theta_{0j}\}_{j=1, \, \ldots, \, p}$ based on the following augmented moment condition:
\begin{equation}
    \theta_{0j} = \E[P]{m_j(W, \gamma_{0j}) + \alpha_{0j}(W)(Y - \gamma_{0j}(W))}, \quad j = 1, \, \ldots, \, p, \label{eq:drmoment}
\end{equation} For the ATE examples, this is the augmented inverse propensity-weighted estimator (AIPW) of \citet{robins1994estimation}. The addition of the term $\alpha_{0j}(W)(Y - \gamma_{0j}(W))$ can be seen as ``debiasing'' the moment condition \eqref{eq:moment}, since it makes it robust to estimation errors in $\ghh_{0j}$ and $\ahh_{0j}$ in a sense made explicit by the proposition below.

\begin{proposition}[Neyman orthogonality and double robustness] \label{prop:neyman} Let $Z = (Y, W)$, and $\psi_j(Z, \theta, \gamma, \alpha)$ denote the augmented score: \[\psi_j(Z, \theta, \gamma, \alpha) = m_j(W, \gamma) + \alpha(W)(Y - \gamma(W)) - \theta.\] We have: \begin{enumerate}[(i)]
    \item (Neyman orthogonality) The Gateaux derivative maps of $\E[P]{\psi_j(Z, \theta, \gamma, \alpha)}$ with respect to $\gamma$ and $\alpha$ are 0 at $(\theta_{0j}, \gamma_{0j}, \alpha_{0j})$:
    \begin{align*}
      \partial_r \E[P]{\psi_j(Z, \theta_{0j}, \gamma_{0j} + r(\gamma - \gamma_{0j}), \alpha_{0j})} |_{r = 0} &= 0 \quad \text{for all } \gamma \in \Gamma,\\ \partial_r \E[P]{\psi_j(Z, \theta_{0j}, \gamma_{0j}, \alpha_{0j}  + r(\alpha - \alpha_{0j}))} |_{r = 0} & = 0 \quad \text{for all } \alpha \in \Gamma.
    \end{align*}

    \item (Double robustness) Moreover, \[\E[P]{\psi_j(Z, \theta_{0j}, \gamma, \alpha)} = - \E[P]{(\alpha(W) - \alpha_{0j}(W))(\gamma(W) - \gamma_{0j}(W))}\] so that the augmented score is mean zero for all $\alpha \in \Gamma$ whenever $\gamma = \gamma_{0j}$, or for all $\gamma \in \Gamma$ whenever $\alpha = \alpha_{0j}$.
\end{enumerate}
\end{proposition}
\begin{proof}
See, e.g., \citet{chernozhukov2018automatic}.
\end{proof}

\begin{remark}\label{rem:tradeoff}
\cref{prop:neyman} (ii) hints at a trade-off between the quality of the estimates for $\gamma_{0j}$ and $\alpha_{0j}$. In situations where $\gamma_{0j}$ can be estimated very well, it may be possible to achieve $\sqrt{n}$-convergence and asymptotic normality even when the rate of convergence of $\ahh_j$ is slow, and vice versa.
\end{remark}

Consider a simultaneous confidence band for $\{\theta_{0j}\}_{j=1, \, \ldots, \, p}$ constructed as $\times_{j=1}^p [\thh_j \mp n^{-1/2}\shh_jc_{\alpha}]$, where the point estimates $\{\thh_{j}\}_{j=1, \, \ldots, \, p}$ are based on an empirical analog of \eqref{eq:drmoment}, \[\thh_{0j} = \En{m_j(W, \ghh_{j}) + \ahh_{j}(W)(Y - \ghh_{j}(W))}, \quad j = 1, \, \ldots, \, p,\]
and $\shh_{j}^2$ is an estimator of $\sigma_j^2 = \mathrm{Var}(\sqrt{n}(\thh_{j} - \theta_{0j}))$. The critical value $c_{\alpha}$ will be chosen such that \[\PP{\max_{1 \leq j \leq p} \sqrt{n}\frac{|\thh_{j} - \theta_{0j}|}{\shh_j} \leq c_\alpha} \approx 1 - \alpha,\] using a normal approximation for the sup-$t$-statistic. The goal of this section is to provide finite-sample guarantees for this normal approximation in the form of a bound on the Kolmogorov distance between the finite-sample distribution of the sup-$t$-statistic and $\max_{1 \leq j \leq p} Z_j$ for a suitable jointly Gaussian distribution $(Z_1, \ldots, Z_p) \sim \mathcal{N}(0, \Sigma)$: \[\sup_{t \in \RR} \left| \PP{\max_{1 \leq j \leq p} \sqrt{n}\frac{(\thh_j - \theta_{0j})}{\sigma_j} \leq t} - \P{\max_{1 \leq j \leq p} Z_j \leq t}\right|.\]
Towards that goal, we list a set of sufficient regularity conditions in the following assumptions.

\begin{assumption}[Moment conditions]\label{ass:mom}
Suppose the following moment conditions hold:
\begin{enumerate}[(i)]
    \item (Bounded heteroskedasticity of the outcome) $\E[P]{(Y - \gamma_{0j}(W))^2 \mid W} \leq \Bar{\sigma}$ for all $j = 1, \, \ldots, \, p$.

    \item (Variance bounded away from 0) Let $\overline{\psi}_{0j}(Z) = m_j(W, \gamma_{0j}) + \alpha_{0j}(Y - \gamma_{0j}(W)) - \theta_{0j}$ denote the oracle score (that is, the score evaluated at the true value of the parameters). We assume $\sigma_j^2 = \E[P]{\overline{\psi}_{0j}(Z)^2} \geq \sigma_{\min}^2 > 0$ for all $j = 1, \, \ldots, \, p$. Moreover, let $\Sigma$ denote the correlation matrix of the $\psi_{0j}(Z)$, with $(j,k)$-th entry given by \[\Sigma_{jk} = \E[P]{\frac{\overline{\psi}_{0j}(Z) \overline{\psi}_{0k}(Z)}{\sigma_j \sigma_k}} \quad j,k = 1, \, \ldots, \, p.\] We assume that the smallest eigenvalue of $\Sigma$ is bounded below by some $\lambda_{\min} \geq 0$.

    \item (Higher-order moments) For some $q \geq 4$, there exists $b_n < \infty$ such that \[\Vert \max_{1 \leq j \leq p} |\overline{\psi}_{0j}(Z)/\sigma_j|\Vert_{P,q} \leq b_n,\] and, for all $j = 1, \, \ldots, \, p$, \[\E[P]{(\overline{\psi}_{0j}(Z)/\sigma_j)^4} \leq b_n^2.\]
\end{enumerate}
\end{assumption}

\begin{remark}[On the eigenvalue condition] The assumption that the minimum eigenvalue of $\Sigma$ is bounded below allows us to obtain nearly-optimal rates with respect to the sample size in the normal approximation we use \citep{chernozhukov2021nearly}. In practice, it implies that the identifying moments do not become perfectly correlated as the number of parameters grows. This restriction precludes certain applications, e.g., using a grid of $p$ points to approximate the CDF of an outcome, with the grid becoming dense asymptotically. The case where there can possibly be a continuum of parameters will be covered in the next section.
\end{remark}

\begin{assumption}[Nuisance parameters]\label{ass:nuisance} Suppose the following:
\begin{enumerate}[(i)]
    \item (RMSE convergence rates) We have $\Vert \ghh_{j} - \gamma_{0j} \Vert_{P,2} \leq \R(\ghh)$ and $\Vert \ahh_{j} - \alpha_{0j} \Vert_{P,2} \leq \R(\ghh)$ for all $j = 1, \, \ldots, \, p$.
    \item (Boundedness of Riesz Representer) We have $\Vert \alpha_{0j} \Vert_{\infty} \leq \Bar{\alpha}$, $\Vert \ahh_{j} \Vert_{\infty} \leq \Bar{\alpha}$ for all $j = 1, \, \ldots, \, p$.
    \item (Envelope and entropy conditions) The class of functions
    \begin{align*}
\F = \{ (Y,W) \mapsto & (m_j(W, \gamma) + \alpha(W)(Y - \gamma(W)) - m_j(W, \gamma_{0j}) - \alpha_{0j}(W)(Y - \gamma_{0j}(W)) \, : \\ & \Vert \gamma - \gamma_{0j} \Vert_{P,2} \leq \R(\ghh), \Vert \alpha - \alpha_{0j} \Vert_{P,2} \leq \R(\ahh), j = 1, \, \ldots, \, p \}.
\end{align*} is suitably measurable, with a measurable envelope $F \geq \sup_{f \in \F} |f|$ that satisfies $\Vert F \Vert_{P,2+\delta} \leq M_n$ for some $\delta\geq 0$ and some sequence of constants $M_n$. There exist sequences $v_n \geq 1$, $a_n \geq n \vee M_n$, such that the uniform entropy numbers of $\F$ obey \begin{align*}
   \log \sup_Q N(\varepsilon \Vert F \Vert_{Q,2}, \F, \Vert \cdot \Vert_{Q,2}) \leq v_n \log(a_n/\varepsilon), \qquad \text{for all } 0 < \varepsilon \leq 1.
\end{align*}
\end{enumerate}
\end{assumption}

\begin{remark}[On entropy conditions]
The goal of entropy conditions is to control the complexity of the class of functions used to estimate nuisance parameters. On the one hand, this class needs to be rich enough for it to be possible to obtain good MSE convergence rates. On the other hand, if the class is too complex, it may lead to an overfitting bias when using the same data to estimate the nuisance parameters $\gamma_{0j}$, $\alpha_{0j}$ and the target parameter $\theta_{0j}$. An alternative to restricting the entropy of the class of functions considered is using some form of sample splitting, as discussed below in \cref{rem:entropy}.
\end{remark}

The following is one of the main theoretical results of this paper. It provides a bound on the Kolmogorov distance between the finite-sample distribution of the sup-$t$-statistic of the DML estimators and $\max_{1 \leq j \leq p} Z_j$ for the Gaussian limit distribution $(Z_1, \ldots, Z_p) \sim \mathcal{N}(0, \Sigma)$, where $\Sigma$ is as defined in Assumption \ref{ass:mom}.
\begin{theorem}\label{thm:linear}
Suppose Assumptions \ref{ass:linearity}, \ref{ass:mom} and \ref{ass:nuisance} hold. Then,
\[\sup_{t \in \RR} \left| \PP{\max_{1 \leq j \leq p} \sqrt{n}\frac{(\thh_j - \theta_{0j})}{\sigma_j} \leq t} - \P{\max_{1 \leq j \leq p} Z_j \leq t}\right| \leq \varrho(n, p),\] where \begin{align*}
\varrho(n, p) & = C(q) \left\lbrace \frac{b_n (\log p)^{3/2} \log n}{\sqrt{n} \lambda_{\min}} + \frac{b_n^2 (\log p)^{2} \log n}{n^{1-2/q} \lambda_{\min}} + \left[\frac{b_n^q (\log d)^{3q/2-4} \log n \log (pn)}{n^{q/2 - 1} (\lambda_{\min})^{q/2}} \right]^\frac{1}{q-2}\right\rbrace \tag{A}\\ & + \frac{6\sqrt{\log p}}{\sigma_{\min}} \left\lbrace \Delta_{1n} + \Delta_{2n} \right\rbrace \tag{B} \\ & + \frac{c}{\log n}
\tag{C} \\
\Delta_{1n} & = K\left(2+\delta,\frac{c}{3}\right)\Bigg( \left[\left((2+\sqrt{2})\Bar{\alpha} + \sqrt{2} \Bar{Q} \right)\R(\ghh) + \Bar{\sigma}\R(\ahh) \right]\sqrt{3v_n\log(3a_n)} \\ & \quad + 3v_n n^{\frac{1}{2+\delta}-\frac{1}{2}} 5M_n \log(3a_n) \Bigg),  \\
\Delta_{2n} & =  \sqrt{n} \R(\ghh) \R(\ahh),
\end{align*} for any $c > 0$, some constant $C(q) > 0$ depending on $q$, and some constant $K\left(2+\delta,\frac{c}{3}\right) > 0$ that may depend on $\delta$ and $c$.
\end{theorem}
\begin{proof}
The full proof is in Appendix \ref{sec:proof1}. Here we discuss the heuristics, which may help understand each of the terms (A), (B) and (C).

The first step of the proof is a decomposition of $\thh_j$ into an oracle estimator, \[\tbb_j = \En{m(W, \gamma_{0j}) + \alpha_{0j}(W)(Y - \gamma_{0j}(W))}\] (i.e., the sample average of the augmented moment condition \eqref{eq:drmoment} if $\gamma_{0j}$ and $\alpha_{0j}$ were known), and a deviation from that oracle estimator, $\mathbb{E}_n[m(W, \ghh_j) + \ahh_{j}(W)(Y - \ghh_{j}(W)) - m(W, \gamma_{0j}) - \alpha_{0j}(W)(Y - \gamma_{0j}(W))]$.

On the one hand, the oracle estimator satisfies a high-dimensional version of a Berry-Esseen type of inequality, which allows us to quantify the Kolmogorov distance between its finite-sample distribution and the distribution of the corresponding multivariate normal distribution. In particular, we use the nearly-optimal rates in \citet{chernozhukov2021nearly}, which yield term (A). We can obtain refinements on this term by assuming stronger conditions on the higher-order moments of $\overline{\psi}_{0j}(Z)$, as discussed in \cref{rem:tails}.

On the other hand, the deviation from the oracle estimator can be bound using empirical process techniques. In the Appendix, we show that, with probability no more than $c/\log n$ (C), the deviation is upper bounded by $\Delta_{1n} + \Delta_{2n}$ (B). Improvements on (B) can be obtained by using some form of sample splitting, as discussed in \cref{rem:entropy}.
\end{proof}

\begin{remark}[Stronger tail conditions]\label{rem:tails}
We could obtain a simpler bound for term (A) by assuming stronger tail conditions than the ones in \cref{ass:mom} (iii), as made clear in \cref{lem:gaus}, which collects results from \citet{chernozhukov2021nearly}. In particular,
\begin{enumerate}[(i)]
    \item Assuming that $\overline{\psi}_{0j}(Z)$ is sub-Gaussian with Orlicz norm upper-bounded by $b_n$, (A) could be replaced by:
    \[C\left\lbrace \frac{b_n (\log p)^{3/2} \log n}{\sqrt{n} \lambda_{\min}} + \frac{b_n^2 (\log p)^{2}}{\sqrt{n\lambda_{\min}}} \right\rbrace\] for an absolute constant $C > 0$.

    \item Assuming that $\overline{\psi}_{0j}(Z)$ is almost-surely bounded by $b_n$, (A) could be replaced by:
    \[C\frac{b_n (\log p)^{3/2} \log n}{\sqrt{n} \lambda_{\min}}\] for an absolute constant $C > 0$.
\end{enumerate}
These stronger assumptions may be satisfied in certain economic applications, for example if outcomes are binary or naturally bounded (e.g., hours worked in a labor supply example).
\end{remark}

\begin{remark}[Removing entropy conditions by sample splitting]\label{rem:entropy}
The role of the entropy conditions in \cref{ass:nuisance} is to prevent overfitting bias, due to the same data being used in the nuisance parameter estimators $\ghh_j$, $\ahh_j$ and the estimator of the target parameter $\thh_j$. Another way to overcome this overfitting problem is, as discussed in \citet{chernozhukov2018double} or \citet{newey2018cross}, to use sample splitting.
\begin{enumerate}
    \item With pure sample splitting, observations 1 to $n$ are divided randomly into $L$ data folds of roughly equal size, $I_\ell$, $\ell = 1, \, \ldots, L$. For a given $\ell$, estimators $\ghh_{j\ell}$, $\ahh_{j\ell}$ of $\gamma_{0j}$ and $\alpha_{0j}$ are constructed using the data \textit{not} in $\ell$, $I_{\ell}^c$. An estimator for $\theta_{0j}$ is then constructed as: \[\thh_{j} = \frac{1}{n} \sum_{\ell = 1}^L \sum_{i \in I_\ell} [m_j(W_i, \ghh_{j\ell}) + \ahh_{j\ell}(W_i)(Y_i - \ghh_{j\ell})].\]
    In that case, within the $\ell$-th fold and after conditioning on $I_{\ell}^c$, the summands are i.i.d., and so we can set $v_n = 1$, $a_n = e$ in \cref{ass:nuisance} (iii).

    \item An alternative is to use a ``dirty'' version of sample splitting, in which data in the $\ell$-th fold is used to \textit{select} amongst a finite set of estimators $\{(\ghh^{(1)}_{j\ell}, \ahh^{(1)}_{j\ell}), \ldots, (\ghh^{(r)}_{j\ell}, \ahh^{(r)}_{j\ell})\}$ trained on data \textit{not} in $\ell$. In that case, because a covering number for a finite class of functions is at most its cardinality, we can set $v_n = 1$, $a_n = e \vee r$. \qedhere
\end{enumerate}
\end{remark}

Finally, the asymptotic validity of the simultaneous confidence band follows as a corollary of \cref{thm:linear} under two additional assumptions. This is not a new result, as it could be seen as a particular case of \citealp{belloni2018uniformly}, but we present it here for completeness.

\begin{assumption}[Consistent variance estimation]\label{ass:var1} Suppose that we have an estimator $\shh_j$ of $\sigma_j$ for all $j = 1, \, \ldots, \, p$ such that $\max_{1\leq j \leq p} (\shh_j/\sigma_j) \overset{p}{\rightarrow} 1$.
\end{assumption}

\begin{assumption}[Growth conditions]\label{ass:growth1} Suppose the following growth conditions as $n, p \rightarrow \infty$:
\begin{enumerate}[(i)]
    \item (Nuisance parameters converge fast enough) $\sqrt{\log (p) n} \R(\ghh)\R(\ahh) \rightarrow 0$.
    \item (Complexity characteristics do not grow too fast) \[\sqrt{\log (p) v_n\log(a_n)} [\R(\ghh)\vee\R(\ahh)] \rightarrow 0  \quad \text{and} \quad \sqrt{\log (p)} v_n n^{\frac{1}{2+\delta} - \frac{1}{2}}M_n \log(a_n) \rightarrow 0.\]
    \item (Moment bounds do not grow too fast)
    \[\frac{b_n (\log p)^{3/2} \log n}{\sqrt{n} \lambda_{\min}} + \frac{b_n^2 (\log p)^{2} \log n}{n^{1-2/q} \lambda_{\min}} + \left[\frac{b_n^q (\log d)^{3q/2-4} \log n \log (pn)}{n^{q/2 - 1} (\lambda_{\min})^{q/2}} \right]^\frac{1}{q-2} \rightarrow 0.\]
\end{enumerate}
\end{assumption}

\begin{remark}
As we pointed out in \cref{rem:tradeoff}, there is a tradeoff between the RMSE convergence rate of $\ghh$ and of $\ahh$, which is made explicit in (i). In the case of a single parameter, $p = 1$, a sufficient condition for (i) is that both $\ghh$ and $\ahh$ converge faster than $n^{-1/4}$, a rate that is typically attainable by non-parametric estimators \citep{chernozhukov2018double}. Note that the dependence of $p$ in the bound is only logarithmic, allowing for very high dimensional cases (potentially $p \gg n$).
\end{remark}

\begin{corollary}[Validity of the simultaneous confidence band]
Under Assumptions \ref{ass:linearity} to \ref{ass:growth1}, we have \[\PP{\thh_j - n^{-1/2}\shh_j c_{\alpha} \leq \theta_{0j} \leq \thh_j + n^{-1/2}\shh_j c_{\alpha}, \, \forall 1 \leq j \leq p} \rightarrow 1 - \alpha.\]
\end{corollary}

\section{Continua of Parameters}
In this section, we consider the same setting as \citet{belloni2017program}. Again, suppose we have access to an i.i.d. sample $\{W_i\}_{i=1}^n$ from a probability distribution $P$ on $\mathcal{W}$. Now, we are interested in  constructing a simultaneous confidence band for the set of (scalar) parameters $\{\theta_{0u}\}_{u \in \U}$, satisfying the following moment condition:
\begin{equation}
    \E[P]{\psi_u(W, \theta_{0u}, \eta_{0u})} = 0, \quad u \in \U \label{eq:momentcont}
\end{equation}
for a possibly uncountable set $\U$, where $\eta_{0u} \in \TT_u$ is the unknown true value of an infinite-dimensional nuisance parameter. Here we also assume that an ML estimator $\ehh_u$ of $\eta_{0u}$ can be obtained from the same data.

Below we give a leading example where this framework may be appropriate.

\begin{example}[Distributional treatment effects]\label{ex:4} Consider a setting where we want to evaluate the effect of a binary treatment $D \in \{0,1\}$ on an outcome $Y$. We work with observational data and, as in Examples $\ref{ex:1}$ and $\ref{ex:2}$, we suppose that we have access to a rich enough set of controls $X$ such that an unconfoundedness assumption holds.

In some cases, features of the marginal distributions of potential outcomes beyond the mean may be of interest. For example, policymakers may care about how the treatment impacts inequality. In other cases, economic theory will make predictions about how different regions of the outcome distribution should be affected by the treatment, so looking at features other than the mean can be used to probe or validate the theory.

Under the unconfoundedness assumption, the marginal distribution of $Y(d)$ can be identified as \[\theta_{0u} = F_{Y(d)}(u) = \E[P]{\gamma_{0u}(d, X)},\] where $\gamma_{0u}(d, x) = F_{Y}(u \mid D = d, X = x) = \PP{Y \leq u \mid D, X}$ is the conditional distribution of $Y$ given $D$ and $X$ at point $u$. A non-parametric estimator of $\gamma_{0u}$ may be constructed using different techniques, for example, distribution regression \citep{chernozhukov2013inference}.

Having access to estimates of $\{\theta_{0u}\}_{u \in \U}$ for a suitable range $\U$ allows to construct many interesting statistics, such as quantile treatment effects, Gini indices, and Oaxaca-Blinder type of decompositions.
\end{example}

Again, our goal in this section is to give finite-sample guarantees for a simultaneous confidence band for $\{\theta_{0u}\}_{u \in \U}$ using a set of DML point estimates $\{\thh_{u}\}_{u \in \U}$ based on an empirical analog of \eqref{eq:momentcont}. We assume that each $\theta_{0u} \in \Theta_u$ for some $\Theta_u \subset \RR$, and that we can find, for each $u \in \U$, an approximate solution to the empirical analog of \eqref{eq:momentcont}, i.e., a $\thh_u$ such that
\begin{equation}
    \En{\psi_u(W, \thh_u, \ehh_u)} \leq n^{-1/2}\epsilon_n.
\end{equation}
for some sequence of $\epsilon_n > 0$ such that $n^{-1/2}\epsilon_n \rightarrow 0$ as $n \rightarrow \infty$. Again, we want to choose a critical value $c_{\alpha}$ such that \[\PP{\sup_{u \in \U} \sqrt{n}\frac{|\thh_{u} - \theta_{0u}|}{\sigma_u} \leq c_\alpha} \approx 1 - \alpha,\] using a normal approximation for the sup-$t$-statistic. As in the previous section, our objective is to provide finite-sample guarantees for this normal approximation in the form of a bound on the Kolmogorov distance between the finite-sample distribution of the sup-$t$-statistic and the supremum of a suitable Gaussian process. We begin by giving some sufficient regularity conditions towards this result.

\begin{assumption}[Moment problem]\label{ass:momentcont}
For all $u \in \U$, the following conditions hold.
\begin{enumerate}[(i)]
    \item The true parameter satisfies $\theta_{0u} \in \mathrm{int} \, \Theta_u$.
    \item The map $(\theta, \eta) \mapsto \E[P]{\psi_u(W, \theta, \eta)}$ is twice continuously Gateaux-differentiable on $\Theta_{u} \times \TT_u$.
    \item (Neyman orthogonality) For $\Bar{r} \in [0,1)$, let $D_{\Bar{r}u}[\eta - \eta_{0u}]$ denote the Gateaux derivative map of $\E[P]{\psi_u(W, \theta, \eta)}$ with respect to $\eta$ at $(\theta_{0u}, \eta_{0u})$ in the direction $\eta - \eta_{0u}$, \[D_{\Bar{r}u}[\eta - \eta_{0u}] = \partial_r \E[P]{\psi_u(W, \theta_{0u}, \eta_{0u} + r(\eta - \eta_{0u}))} |_{r = \Bar{r}}.\] Then, $D_{0u}[\eta - \eta_{0u}] = 0$ for all $\eta \in \TT_u$.
    \item (Bounded derivatives with respect to $\theta$) Let $J_{0u} = \partial_{\theta} \E[P]{\psi_u(W, \theta, \eta_{0u})} |_{\theta = \theta_{0u}}$. Then $c_0 \leq |J_{0u}| \leq C_0$. Moreover, $|\partial_{\theta} \E[P]{\psi_u(W, \theta, \eta_{0u})}| > c_1$ for all $\theta \in \Theta_u$.
    \item (Lipschitz-continuity at the true parameters) For all $\theta \in \Theta_u$ and $\eta \in \TT_u$, \[\E[P]{(\psi_u(W, \theta, \eta) - \psi_u(W, \theta_{0u}, \eta_{0u}))^2} \leq C_0 (|\theta-\theta_{0u}|\vee\Vert\eta - \eta_{0u}\Vert_{P,2})^\omega.\]
    \item (Bounded derivatives with respect to $\eta$) For all $r \in [0,1)$, $\theta \in \Theta_u$ and $\eta \in \TT_u$, \[|\partial_r \E[P]{\psi_u(W, \theta, \eta_{0u} + r(\eta - \eta_{0u}))}| \leq B_{1n}\Vert\eta - \eta_{0u}\Vert_{P,2}. \]
    \item (Bounded second derivatives) For all $r \in [0,1)$, $\theta \in \Theta_u$ and $\eta \in \TT_u$,
    \[|\partial^2_r \E[P]{\psi_u(W, \theta_{0u} + r(\theta - \theta_{0u}), \eta_{0u} + r(\eta - \eta_{0u}))}| \leq B_{2n}(|\theta-\theta_{0u}|\vee\Vert\eta - \eta_{0u}\Vert_{P,2})^2. \]
\end{enumerate}
\end{assumption}

\begin{remark}[On the Neyman orthogonality condition] As opposed to the previous section, here we take Neyman orthonality (iii) as a primitive condition, and hence we assume that $\eta_{0j}$ contains all nuisance parameters needed to make the moment functional Neyman-orthogonal. In the case of a linear, mean-square continuous functional of a regression, the same construction as in \cref{sec:lin} is valid, and so $\eta_{0u} = (\gamma_{0u}, \alpha_{0u})$ for $\alpha_{0u}$ the Riesz representer. This is true, for instance, in \cref{ex:4}. More generally, \citet{chernozhukov2018automatic} discuss how to orthogonalize non-linear functionals, and \citet{chernozhukov2018double}, \citet{belloni2018uniformly} cover many other important models, such as conditional moment restrictions.
\end{remark}

\begin{assumption}[Nuisance parameters]\label{ass:nuiscont} Suppose the following:
\begin{enumerate}[(i)]
    \item (RMSE convergence rates) For all $u \in \U$ we have $\Vert\ehh_u - \eta_{0u}\Vert_{P,2} \leq \R(\ehh)$.
    \item (Envelope and entropy conditions) The class of functions \[\F = \{W \mapsto \psi_u(W, \theta, \eta) \, : \, \theta \in \Theta_u, \eta \in \TT_u, u \in \U\}\] is suitably measurable, with a measurable envelope $F \geq \sup_{f \in \F} |f|$ that satisfies $\Vert F \Vert_{P,2+\delta} \leq M_n$ for some $\delta\geq 0$ and some sequence of constants $M_n$. There exist sequences $v_n \geq 1$, $a_n \geq n \vee M_n$, such that the uniform entropy numbers of $\F$ obey \begin{align*}
   \log \sup_Q N(\varepsilon \Vert F \Vert_{Q,2}, \F, \Vert \cdot \Vert_{Q,2}) \leq v_n \log(a_n/\varepsilon), \qquad \text{for all } 0 < \varepsilon \leq 1.
    \end{align*} Finally, for all $f \in \F$, we have $c_0 \leq \Vert f \Vert_{P,2} \leq C_0$.
\end{enumerate}
\end{assumption}

\begin{assumption}[Entropy and moments of the score at the truth]\label{ass:scorecont} Suppose $\sigma_u^2 = J_{0u}^{-2} \mathrm{E}_P[\psi_u(W, \allowbreak \theta_{0u}, \eta_{0u})^2] \geq C_0^{-2}c_0^2$, and let $\overline{\psi}_{0u}(W) = -(\sigma_{u}J_{0u})^{-1}  \psi_u(W, \allowbreak \theta_{0u}, \eta_{0u})$ denote the re-scaled score evaluated at the true values $(\theta_{0u}, \eta_{0u})$. Then:
\begin{enumerate}[(i)]
    \item (Envelope and entropy conditions) The class of functions \[\F_0 = \{W \mapsto \overline{\psi}_{0u}(W) \, : \,u \in \U\}\] is suitably measurable, with a measurable envelope $F_0 \geq \sup_{f \in \F_0} |f|$ that satisfies $\Vert F_0 \Vert_{P,q} \leq b_n$ for some $q\geq 4$ and some sequence of constants $b_n$. There exist sequences $V_n \geq 1$, $A_n \geq n$, such that the uniform entropy numbers of $\F_0$ obey \begin{align*}
   \log \sup_Q N(\varepsilon \Vert F_0 \Vert_{Q,2}, \F_0, \Vert \cdot \Vert_{Q,2}) \leq V_n \log(A_n/\varepsilon), \qquad \text{for all } 0 < \varepsilon \leq 1.
    \end{align*}

    \item (Moments) For all $f \in \F_0$ and $k = 3,4$, $\E[P]{|f(W)|^k} \leq C_0 b_n^{k-2}$.
\end{enumerate}
\end{assumption}

The following theorem is the second main result of this paper. Again, it provides a bound on the Kolmogorov distance between the finite-sample distribution of the sup-$t$-statistic of the DML estimators and its limiting distribution. In the statement of the theorem, $G_P$ denotes a tight mean-zero Gaussian process indexed by the class of functions in \cref{ass:scorecont}, with covariance function $\E{G_P [\overline{\psi}_{0u}] G_P [\overline{\psi}_{0u'}] } = \E[P]{\overline{\psi}_{0u}(W) \overline{\psi}_{0u'}(W)}$ for all $u,u' \in \U$.

\begin{theorem}\label{thm:cont}
Suppose Assumptions \ref{ass:momentcont}, \ref{ass:nuiscont} and \ref{ass:scorecont} hold. Then,
\begin{multline*}
\sup_{t \in \RR} \left| \PP{\sup_{u \in \U} \sqrt{n}\frac{(\thh_u - \theta_{0u})}{\sigma_u} \leq t} - \P{\sup_{u \in \U} G_P [\overline{\psi}_{0u}] \leq t}\right| \leq \\ \kappa r_{1n} \left(\chi \sqrt{V_n \log(A_n b_n)} + \sqrt{1 \vee \log(1 / r_{1n})}\right) + r_{2n},
\end{multline*}
where $\kappa, \chi > 0$ are universal constants,
\[r_{1n} = c_0^{-1}\epsilon_n + \Delta_{1n} + \Delta_{2n} + \Delta_{3n}\] for:
\begin{align*}
\Delta_{1n} & = C_0^{-1}K\left(2+\delta,\frac{c}{2}\right)\Bigg(
 \sqrt{C_0} [\R^\vee (\ehh)]^{\omega/2}\sqrt{2v_n\log(2a_n)} + 2v_n n^{\frac{1}{2+\delta}-\frac{1}{2}} 2M_n \log(2a_n) \Bigg). \\
\Delta_{2n} & = C_0^{-1} \tfrac{1}{2} \sqrt{n} B_{2n} [\R^\vee (\ehh)]^2. \\
\Delta_{3n} & = \frac{b_nL_n}{\gamma^{1/2} n^{1/2 - 1/q}} + \frac{(b_n)^{1/2}L_n^{3/4}}{\gamma^{1/2} n^{1/4}} + \frac{(b_nL_n^2)^{1/3}}{\gamma^{1/3} n^{1/6}} \\
\R^\vee (\ehh) & = \Big\lbrace c_1^{-1} n^{-1/2}\epsilon_n + c_1^{-1} n^{-1/2}K\left(2 + \delta, \frac{c}{2}\right) \left(C_0 \sqrt{v_n \log(a_n)} + v_n n^{\frac{1}{2+\delta} - \frac{1}{2}}M_n\log(a_n) \right) \\
& + c_1^{-1} B_{1n}\R(\ehh)\Big\rbrace \vee \R(\ehh), \\
L_n & = d(q) V_n(\log n \vee \log(A_n b_n)),
\end{align*}
and \[r_{2n} = D(q)\left(\gamma + \log n / n\right) + c/\log n,\] where $d(q),D(q)$ are constants depending only on $q$, for any $c>0$ and $\gamma \in (0,1)$.
\end{theorem}

\begin{proof}
The full proof is in Appendix \ref{sec:proof2}. As above, we discuss the heuristics here. First, $\R^\vee (\ehh)$ is the maximum of two objects: the rate $\R(\ehh)$ for $\ehh$ and a preliminary rate for $\thh$ (an upper bound for how far $\thh$ can be from $\theta_0$ based only on the smoothness conditions). Notice that this step becomes unnecessary whenever the moment function $\psi_{u}$ is linear in $\theta$, which will be the case in many applications, including \cref{ex:4}.

Second, the Kolmogorov distance-based statement of the theorem is related by Lemma \ref{lem:gausemp2} to another kind of approximation, of the form: \[\PP{\left| \sup_{u \in \U} \sqrt{n}\frac{(\thh_u - \theta_{0u})}{\sigma_u} - G_P [\overline{\psi}_{0u}] \right| > r_{1n}} \leq r_{2n}.\]
As an intermediate step, we first approximate $\sqrt{n}(\thh_u - \theta_{0u})/\sigma_u$ by the empirical process on the re-scaled oracle score, $\Gn[\Bar{\psi}_{0u}]$. With high probability, the distance between these two objects is bounded by $\Delta_{1n} + \Delta_{2n}$. The first term quantifies the size of the deviation between $\psi_u(W, \thh, \ehh)$ and $\psi_u(W, \theta_{0u}, \eta_{0u})$ when $\psi_u(W, \thh, \ehh)$ is in the class of functions $\F$. The second term bounds the error that we incur by linearizing the score. In particular, if the score is linear in both $\theta$ and $\eta$, this term can be ignored.

Finally, the term $\Delta_{3n}$ arises when approximating the supremum of the empirical process on the oracle score, $\Gn[\Bar{\psi}_{0u}]$, by the supremum of the corresponding Gaussian process, $G_P[\Bar{\psi}_{0u}]$, and it is a consequence of \cref{lem:gausemp1}.
\end{proof}

As in the previous section, we give two additional conditions for the asymptotic validity of the simultaneous confidence band. This is also not a new result, but was shown in \citealp{belloni2017program, belloni2018uniformly}. We present it below for completeness.

\begin{assumption}[Consistent variance estimation]\label{ass:var2} Suppose that we have an estimator $\shh_u$ of $\sigma_u$ for all $u \in \U$ such that $\sup_{u \in \U} (\shh_u/\sigma_u) \overset{p}{\rightarrow} 1$.
\end{assumption}

\begin{assumption}[Growth conditions]\label{ass:growth2} Suppose the following growth conditions as $n \rightarrow \infty$:
\begin{enumerate}[(i)]
    \item (Nuisance parameters converge fast enough) $\sqrt{n} [\R(\ehh)]^2 \rightarrow 0$.
    \item (Complexity characteristics do not grow too fast) \[\sqrt{v_n\log(a_n)} [\R^\vee (\ehh)]^{\omega/2}, \quad v_n n^{\frac{1}{2+\delta} - \frac{1}{2}}M_n \log(a_n) \rightarrow 0 \quad \text{and} \quad r_{1n} \sqrt{V_n \log(A_n b_n)} \rightarrow 0.\]
    \item (Moment bounds do not grow too fast)
    \[\frac{b_nL_n}{\gamma^{1/2} n^{1/2 - 1/q}} + \frac{(b_n)^{1/2}L_n^{3/4}}{\gamma^{1/2} n^{1/4}} + \frac{(b_nL_n^2)^{1/3}}{\gamma^{1/3} n^{1/6}} \rightarrow 0.\]
\end{enumerate}
\end{assumption}

\begin{corollary}[Validity of the simultaneous confidence band]
Under Assumptions \ref{ass:momentcont} to \ref{ass:growth2}, we have \[\PP{\thh_u - n^{-1/2}\shh_u c_{\alpha} \leq \theta_{0u} \leq \thh_u + n^{-1/2}\shh_u c_{\alpha}, \, \forall u \in \U} \rightarrow 1 - \alpha.\]
\end{corollary}

\section{Conclusions}
In many applications, researchers are interested in the causal impact of a treatment or policy that was not randomly assigned. Inference in such non-experimental settings is still possible by controlling for a set of covariates, conditional on which the treatment becomes plausibly exogenous. In modern settings, DML offers an alternative way to leverage a large number of potential controls with regularization and model selection, which does not bias the estimates of the target parameters thanks to a Neyman orthogonality condition. Often, we want to make joint inference on a high-dimensional set of parameters, such as the ATE of many treatments or combinations thereof, the ATE of a treatment on many outcomes, or effects on the entire marginal distribution of the potential outcomes. In this paper we have complemented existing asymptotic results for high-dimensional DML \citep{belloni2017program, belloni2018uniformly} with finite-sample guarantees. These finite sample guarantees can be useful to applied researchers, as they are informative of how far off the coverage of simultaneous confidence bands can be from the nominal levels.

There is one natural extension of the paper that would be an interesting avenue for future research. Our guarantees are based on the true standard error, $\sigma_j$ or $\sigma_u$, respectively, which in general will not be known and will have to be estimated. For our asymptotic corollaries, we have simply assumed that estimators $\shh_j$ or $\shh_u$ exist. We leave it to future work to provide finite-sample guarantees for variance estimation in the high dimensional case, although we note that such guarantees are available for one-dimensional DML with sample splitting \citep{chernozhukov2021simple}.

\clearpage
\bibliography{references}

\clearpage

\begin{appendices}
\section{Proofs}
\subsection{Proof of \cref{thm:linear}}\label{sec:proof1}
The DML estimator of $\theta_{0j}$ is, for all $j = 1, \, \ldots, \, p$, \[\thh_j = \En{m(W, \ghh_j) + \ahh_j(W)(Y - \ghh_j(W))}.\]
Decompose it as in \citet{chernozhukov2021simple}: \begin{align}\label{eq:decomp}
    \sqrt{n}(\thh_j - \theta_{0j}) & =  \sqrt{n}(\tbb_j - \theta_{0j}) + A_{jn} + B_{jn} + C_{jn} + D_{jn},
\end{align}
where: \begin{enumerate}[(i)]
    \item $\tbb_j = \En{m(W, \gamma_{0j}) + \alpha_{0j}(W)(Y - \gamma_{0j}(W))}$ is the oracle estimator.

    \item $A_{jn} = \Gn [m(W, \ghh_j - \gamma_{0j}) + \alpha_{0j}(W)(\ghh_j(W) - \gamma_{0j}(W))]$,

    \item $B_{jn} = \Gn [(\ahh_j(W) - \alpha_{0j}(W))(Y - \gamma_{0j}(W))]$,

    \item $C_{jn} = \Gn [-(\ahh_j(W) - \alpha_{0j}(W))(\ghh_j(W) - \gamma_{0j}(W))]$,

    \item $D_{jn} = - \sqrt{n} \, \E[P]{(\ahh_j(W) - \alpha_{0j}(W))(\ghh_j(W) - \gamma_{0j}(W))}$.
\end{enumerate}

The first term in the right-hand side of \eqref{eq:decomp} is an average of random variables. We can quantify, using a multivariate Berry-Essen inequality, how far its finite-sample distribution is from that of a Gaussian random variable. The terms $A_{jn}$, $B_{jn}$ and $C_{jn}$ are more challenging to bound, since they are empirical process of functions that themselves depend on the data through the estimated nuisance parameters $\ghh_j$ and $\ahh_j$. \citet{chernozhukov2021simple} proceed by using cross-fitting --- i.e., splitting the sample in folds and, for each fold, use an estimate of $\gamma_{0j}$ and $\alpha_{0j}$ based on the remaining data. In this paper, we employ empirical process inequalities instead. Finally, the last term, $D_{jn}$, can be bounded using simple moment inequalities.

\paragraph{Bounds on $A_{jn}$, $B_{jn}$ and $C_{jn}$}
We need to obtain high-probability bounds for $A_{jn}, B_{jn}, C_{jn}$ that hold uniformly over $j = 1, \, \ldots, \, p$. To do so, we will apply \cref{lem:max} to the following empirical processes:  \[\max_{1 \leq j \leq p} |A_{jn}| \leq \sup_{f \in \F_A} \Gn[f], \quad \max_{1 \leq j \leq p} |B_{jn}| \leq \sup_{f \in \F_B} \Gn[f], \quad \max_{1 \leq j \leq p} |C_{jn}| \leq \sup_{f \in \F_C} \Gn[f],\]
with the corresponding classes of functions:
\begin{align*}
\F_A = \lbrace & m_j(W, \gamma - \gamma_{0j}) - \alpha_{0j}(W)(\gamma(W) - \gamma_{0j}(W)) \, : \, \Vert \gamma - \gamma_{0j} \Vert_{P,2} \leq \R(\ghh), j = 1, \, \ldots, \, p \rbrace, \\ \F_B = \lbrace & (\alpha(W) - \alpha_{0j}(W))(Y - \gamma_{0j}(W)) \, : \, \Vert \alpha - \alpha_{0j} \Vert_{P,2} \leq \R(\ahh), j = 1, \, \ldots, \, p \rbrace, \\
\F_C = \lbrace & -(\alpha(W) - \alpha_{0j}(W))(\gamma(W) - \gamma_{0j}(W)) \, : \, \Vert \gamma - \gamma_{0j} \Vert_{P,2} \leq \R(\ghh), \\ & \Vert \alpha - \alpha_{0j} \Vert_{P,2} \leq \R(\ahh), j = 1, \, \ldots, \, p \rbrace.
\end{align*}

Recall the class of functions defined in \cref{ass:nuisance},
\begin{align*}
\F = \{ & (m_j(W, \gamma) + \alpha(Y - \gamma) - m_j(W, \gamma_{0j}) - \alpha_{0j}(Y - \gamma_{0j}) \, : \\ & \Vert \gamma - \gamma_{0j} \Vert_{P,2} \leq \R(\ghh), \Vert \alpha - \alpha_{0j} \Vert_{P,2} \leq \R(\ahh), j = 1, \, \ldots, \, p \}.
\end{align*}
We have that $\F_A \subset \F$, $\F_B \subset \F$ and $\F_C \subset \F - \F_A - \F_B$. By assumption, there exists an envelope function $F$ for $\F$, such that $F \geq \sup_{f\in\F} |f|$ with $\Vert F \Vert_{P,2+\delta} \leq M_n$ for some $\delta \geq 0$. Hence, the corresponding envelopes for $\F_A$, $\F_B$ and $\F_C$ are $F$, $F$ and $3F$. By Assumptions \ref{ass:linearity}, \ref{ass:mom} and \ref{ass:nuisance}, the classes $\F_A$, $\F_B$ and $\F_C$ satisfy $\sup_{f \in \F_k} \Vert f \Vert^2_{P,2} \leq \sigma^2_k$ for $k = A, B, C$, where:
\[\sigma_A^2 = 2(\Bar{Q}^2 + \Bar{\alpha}^2) \R(\ghh)^2, \quad \sigma_B^2 = \Bar{\sigma}^2 \R(\ahh)^2, \quad \sigma_C^2 = 4 \Bar{\alpha}^2 \R(\ghh)^2.\]

Finally, the uniform covering entropy of $\F_A$ and $\F_B$ is upper bounded by that of $\F$. For $\F_C$, we can obtain an upper bound using \cref{lem:entr}:
\[
\log \sup_Q N(\varepsilon \Vert 3F\Vert_{Q,2}, \F_C, \Vert \cdot \Vert_{Q,2}) \leq 3v_n \log(3a_n/\varepsilon) \quad \text{ for all } 0 < \varepsilon \leq 1.
\]

By \cref{lem:max}, we conclude that, with probability at least $1 - c/\log n$,
\begin{multline*}
 \max_{1 \leq j \leq p} |A_{jn}| + \max_{1 \leq j \leq p} |B_{jn}| + \max_{1 \leq j \leq p} |C_{jn}| \leq \\ \Delta_{1n} = K\left(2+\delta,\frac{c}{3}\right)\Bigg( \left[\left((2+\sqrt{2})\Bar{\alpha} + \sqrt{2} \Bar{Q} \right)\R(\ghh) + \Bar{\sigma}\R(\ahh) \right]\sqrt{3v_n\log(3a_n)} \\ + \  3v_n n^{\frac{1}{2+\delta}-\frac{1}{2}} 5M_n \log(3a_n) \Bigg).  \end{multline*}

\paragraph{Bound on $D_{jn}$}
By the Cauchy-Schwarz inequality, $|D_{jn}| \leq \Delta_{2n} = \sqrt{n} \R(\ghh) \R(\ahh)$ for each $j = 1, \, \ldots, \, p$. Hence, $\max_{1 \leq j \leq p} |D_{jn}| \leq \Delta_{2n}$.

\paragraph{Normal Approximation}
Combining the two bounds, we have:
\begin{gather*}
\left| \PP{\max_{1 \leq j \leq p} \sqrt{n}\frac{(\thh_j - \theta_{0j})}{\sigma_j} \leq z} - \P{\max_{1 \leq j \leq p} Z_j \leq z}\right| \hspace{10em} \\ \leq \left| \PP{\max_{1 \leq j \leq p} \sqrt{n}\frac{(\tbb_j - \theta_{0j})}{\sigma_j} \leq z + \frac{\Delta_n}{\sigma_{\min}}} - \P{\max_{1 \leq j \leq p} Z_j \leq z + \frac{\Delta_n}{\sigma_{\min}}}\right| \\ + \left| \P{\max_{1 \leq j \leq p} Z_j \leq z + \frac{\Delta_n}{\sigma_{\min}}} - \P{\max_{1 \leq j \leq p} Z_j \leq z}\right| + \frac{c}{\log n},
\end{gather*}
where $\Delta_n = \Delta_{1n} + \Delta_{2n}$ and $\bm Z = (Z_1, \ldots, Z_p) \sim \mathcal{N}(0, \Sigma)$, for $\Sigma$ and $\sigma_{\min}$ defined in \cref{ass:mom}.

By the results of \citet{chernozhukov2021nearly}, which we collect in \cref{lem:gaus}, we have that the first term on the right-hand side of the inequality is bounded above by:
\begin{multline*}
 \sup_{t \in \RR}\left| \PP{\max_{1 \leq j \leq p} \sqrt{n}\frac{(\thh_j - \theta_{0j})}{\sigma_j} \leq t} - \P{\max_{1 \leq j \leq p} |Z_j| \leq t}\right| \leq \\ C(q) \left\lbrace \frac{b_n (\log p)^{3/2} \log n}{\sqrt{n} \lambda_{\min}} + \frac{b_n^2 (\log p)^{2} \log n}{n^{1-2/q} \lambda_{\min}} + \left[\frac{b_n^q (\log d)^{3q/2-4} \log n \log (pn)}{n^{q/2 - 1} (\lambda_{\min})^{q/2}} \right]^\frac{1}{q-2}\right\rbrace,
\end{multline*}
for a constant $C(q) > 0$ depending only on $q$, where the quantities $b_n$, $q$ and $\lambda_{\min}$ are defined in \cref{ass:mom}.

To bound the second term, we invoke an anti-concentration inequality (\cref{lem:antic}). Since the $Z_j$ are centered and unit-variance, we get: \begin{align*}
  \sup_{t \in \RR}\left| \P{\max_{1 \leq j \leq p} Z_j \leq z + \frac{\Delta_n}{\sigma_{\min}}} - \P{\max_{1 \leq j \leq p} Z_j \leq z}\right| & =  \sup_{x \in \RR}\P{\max_{1 \leq j \leq p} Z_j \in \left[x \mp \frac{\Delta_n}{2\sigma_{\min}}\right]} \\ & \leq \frac{6\sqrt{\log p}}{\sigma_{\min}}\Delta_n.
\end{align*}

\subsection{Proof of \cref{thm:cont}}\label{sec:proof2}
As in the proof of \cref{thm:linear}, we begin by approximating $\thh_u$ by an ``oracle'' estimator.
To do that, fix $u \in \U$ and consider a second-order Taylor expansion of the function
\begin{equation}\label{eq:taylor}
  f(r) = \E[P]{\psi_u(W, \theta_{0u} + r(\thh_u - \theta_{0u}), \eta_{0u} + r(\ehh_u - \eta_{0u}))}
\end{equation} around $r = 0$:
\begin{align*}
    f(1) - f(0) & = \E[P]{\psi_u(W, \thh_u, \ehh_u)} - \E[P]{\psi_u(W, \theta_{0u}, \eta_{0u}))} \\ & = J_{0u} (\thh_u - \theta_{0u}) + D_{0u} [\ehh_u - \eta_{0u}] + \tfrac{1}{2}f''(\Bar{r})
\end{align*} for some $r \in [0,1]$. Because $\E[P]{\psi_u(W, \theta_{0u}, \eta_{0u}))} = 0$ and $D_{0u} [\ehh_u - \eta_{0u}] = 0$ because of the Neyman-orthogonality condition, we have \[\sqrt{n}(\thh_{u} - \theta_{0u}) = \sqrt{n} J_{0u}^{-1}\E[P]{\psi_u(W, \thh_u, \ehh_u)} - \sqrt{n}J_{0u}^{-1}\tfrac{1}{2}f''(\Bar{r}).\]
We can further expand the right hand side to obtain:
\begin{equation}\label{eq:exp2}
\sqrt{n}(\thh_{u} - \theta_{0u}) = -J_{0u}^{-1} \Gn [\psi_u(W, \theta_{0u}, \eta_{0u})] + A_{un} + B_{un} + C_{un},
\end{equation}
where:

\begin{enumerate}[(i)]
    \item The first term on the right-hand side is the oracle score,

    \item $A_{un} = \sqrt{n} J_{0u}^{-1} \En{\psi_u(W, \thh_u, \ehh_u)}$, which is related to the approximation error in $\thh_u$,

    \item $B_{un} = -J_{0u}^{-1}\Gn[\psi_u(W, \thh_u, \ehh_u) - \psi_u(W, \theta_{0u}, \eta_{0u})]$, the empirical process controlling the deviation of $\psi_u(W, \thh_u, \ehh_u)$ from $\psi_u(W, \theta_{0u}, \eta_{0u})$,

    \item $C_{un} = -\sqrt{n}J_{0u}^{-1}\tfrac{1}{2}f''(\Bar{r})$ is a linearization error.
\end{enumerate}

\paragraph{Bound on $A_{un}$}
By \cref{ass:momentcont} $\left| J_{0u} \right| \geq c_0$ for all $u \in \U$. Moreover, because $\thh_u$ is an $n^{-1/2}\epsilon_n$-approximate solution to $\En{\psi(W, \theta, \ehh_u)} = 0$ for all $u \in \U$, we have $\sup_{u\in\U} \left| A_{un} \right| \leq c_0^{-1} \epsilon_n$.

\paragraph{Bound on $B_{un}$} We use a similar argument as before. Let
\begin{align*}
 \F_B = \{\psi_u(W, \theta, \eta) - \psi_u(W, \theta_{0u}, \eta_{0u}) \, : \, & |\theta - \theta_{0u}| \leq \R(\thh), \theta \in \Theta_u, \\ & \Vert \eta - \eta_{0u}\Vert_{P,2} \leq \R(\ehh), u \in \U\},
\end{align*}
where $\R(\thh) = \sup_{u \in U} |\thh_{u} - \theta_{0u}|$ is a preliminary rate of convergence for $\{\thh_{u}\}_{u \in \U}$. A bound on this quantity will be obtained below.

We have that $\F_B \subset \F - \F$ for the class $\F$ defined in \cref{ass:nuiscont}. By assumption, there exists an envelope function $F$ for $\mathcal{F}$, such that $F \geq \sup_{f\in\F} |f|$ with $\Vert F \Vert_{P,2+\delta} \leq M_n$ for some $\delta \geq 0$. Hence, the corresponding envelope $\F_B$ is $2F$. Moreover, the uniform covering entropy of $\F_B$ can be upper-bounded by \cref{lem:entr}:
\[
\log \sup_Q N(\varepsilon \Vert 2F\Vert_{Q,2}, \F_B, \Vert \cdot \Vert_{Q,2}) \leq 2v_n \log(2a_n/\varepsilon) \quad \text{ for all } 0 < \varepsilon \leq 1.
\]
Moreover, by the Lipschitzness-continuity condition in \cref{ass:momentcont}, we can bound
\begin{align*}
\sup_{f \in \F_B} \Vert f \Vert_{P,2} & = \sup_{
   |\theta - \theta_{0u}| \leq \R(\thh), \theta \in \Theta_u, \Vert \eta - \eta_{0u}\Vert_{P,2} \leq \R(\ehh), u \in \U} \E[P]{(\psi_u(W, \theta, \eta) - \psi_u(W, \theta_{0u}, \eta_{0u}))^2} \\
   & \leq C_0 \sup_{
   |\theta - \theta_{0u}| \leq \R(\thh), \theta \in \Theta_u, \Vert \eta - \eta_{0u}\Vert_{P,2} \leq \R(\ehh), u \in \U} (|\theta - \theta_{uj}| \vee \Vert \eta - \eta_{0u}\Vert_{P,2})^{\omega} \\
   & = C_0 [\R(\thh) \vee \R(\ehh)]^\omega.
\end{align*}

By \cref{lem:max}, we conclude that, with probability at least $1 - c/(2\log n)$,
\begin{multline*}
 \sup_{u \in U}|B_{un}| \leq c_0^{-1}K\left(2+\delta,\frac{c}{2}\right)\Bigg(
 \sqrt{C_0} [\R(\thh) \vee \R(\ehh)]^{\omega/2}\sqrt{2v_n\log(2a_n)} \\ + \  2v_n n^{\frac{1}{2+\delta}-\frac{1}{2}} 2M_n \log(2a_n) \Bigg).
\end{multline*}

\paragraph{Bound on $C_{un}$}
By the smoothness condition in \cref{ass:momentcont}, we have that
$\sup_{u \in \U} |f''(\Bar{r})| \leq B_{2n} \sup_{u \in \U} (|\thh_{u} - \theta_{0u}| \vee \Vert \eta - \eta_{0u}\Vert_{P,2})^2 \leq B_{2n} [\R(\thh) \vee \R(\ehh)]^2$. Hence, \[\sup_{u \in \U} |C_{un}| \leq c_0^{-1} \tfrac{1}{2} \sqrt{n} B_{2n} [\R(\thh) \vee \R(\ehh)]^2.\]

\paragraph{Preliminary rate} Here we obtain an upper bound for $\R(\thh)$ in terms of $\R(\ehh)$ and other primitive quantities. To that end, consider now a first-order Taylor expansion of \eqref{eq:taylor} around $r = 0$:
\begin{align*}
 f(1) - f(0) & = \E[P]{\psi_u(W, \thh_u, \ehh_u)} - \E[P]{\psi_u(W, \theta_{0u}, \eta_{0u}))} \\ & = J_{\Bar{r} u} (\thh_u - \theta_{0u}) + D_{\Bar{r} u} [\ehh_u - \eta_{0u}]
\end{align*}
for some $\Bar{r} \in [0,1]$.
By the smoothness conditions, $|D_{\Bar{r} u} [\ehh_u - \eta_{0u}]| \leq B_{1n} \Vert \eta - \eta_{0u}\Vert_{P,2} \leq B_{1n}\R(\ehh)$ for all $u \in \U$. Moreover, $|J_{\Bar{r} u}| \geq c_{1}$ for any $\Bar{r} \in [0,1)$ and $u \in U$. Re-arranging, and adding and subtracting $\En{\psi_u(W, \thh_u, \ehh_u)}$ we obtain
\begin{align*}
\sup_{u \in \U}|\thh_u - \theta_{0u}| & \leq c_1^{-1} \left|\En{\psi_u(W, \thh_u, \ehh_u)}\right| + c_1^{-1} n^{-1/2}\sup_{u \in \U} |\Gn \psi(W, \thh_u, \ehh_u)| + c_1^{-1} B_{1n}\R(\ehh) \\
& \leq c_1^{-1} n^{-1/2}\epsilon_n + c_1^{-1} n^{-1/2}\sup_{f \in \F} |\Gn f| + c_1^{-1} B_{1n}\R(\ehh)
\end{align*}
for the class $\F$ defined in \cref{ass:nuiscont}. By \cref{lem:max}, with probability at least $1 - c/(2\log n)$
\[\sup_{f \in \F} |\Gn f| \leq K\left(2 + \delta, \frac{c}{2}\right) \left(C_0 \sqrt{v_n \log(a_n)} + v_n n^{\frac{1}{2+\delta} - \frac{1}{2}}M_n\log(a_n) \right).\]

Below, we denote:
\begin{multline*}
 \R^\vee (\ehh) = \Big\lbrace c_1^{-1} n^{-1/2}\epsilon_n + c_1^{-1} n^{-1/2}K\left(2 + \delta, \frac{c}{2}\right) \left(C_0 \sqrt{v_n \log(a_n)} + v_n n^{\frac{1}{2+\delta} - \frac{1}{2}}M_n\log(a_n) \right) \\ + c_1^{-1} B_{1n}\R(\ehh)\Big\rbrace \vee \R(\ehh).
\end{multline*}
\begin{align*}
\Delta_{1n} & = C_0^{-1}K\left(2+\delta,\frac{c}{2}\right)\Bigg(
 \sqrt{C_0} [\R^\vee (\ehh)]^{\omega/2}\sqrt{2v_n\log(2a_n)} + 2v_n n^{\frac{1}{2+\delta}-\frac{1}{2}} 2M_n \log(2a_n) \Bigg). \\
\Delta_{2n} & = C_0^{-1} \tfrac{1}{2} \sqrt{n} B_{2n} [\R^\vee (\ehh)]^2.
\end{align*}

\paragraph{Normal Approximation} To obtain bounds on the Kolmogorov distance between the finite-sample distribution of $\{\thh_{u}\}_{u \in \U}$ and that of the corresponding Gaussian process $\{Z_u\}$ we will use the results of \citet{chernozhukov2014gaussian}, combining \cref{lem:gausemp1} and \cref{lem:gausemp2}.

Let $Z = \sup_{u \in \U} \sqrt{n}\sigma_{u}^{-1}(\thh_{u} - \theta_{0u})$, $\overline{Z} = \sup_{u \in \U} \Gn [\overline{\psi}_{0u}]$, where $\overline{\psi}_{0u}(W) = -(\sigma_{u}J_{0u})^{-1}  \psi_u(W, \allowbreak \theta_{0u}, \eta_{0u})$ is the re-scaled score at the true values $(\theta_{0u}, \eta_{0u})$, and $\widetilde{Z} = \sup_{u \in \U} G_P [\overline{\psi}_{0u}]$, for a tight mean-zero Gaussian process $G_P$ with covariance function $\E{G_P [\overline{\psi}_{0u}] G_P [\overline{\psi}_{0u'}] } = \E[P]{\overline{\psi}_{0u}(W) \overline{\psi}_{0u'}(W)}$. Notice that, by construction, $\E[P]{\overline{\psi}_{0u}(W)^2} = 1$ for all $u \in \U$.

By \cref{lem:gausemp1} and \cref{ass:scorecont}, we have, for any $\gamma \in (0,1)$, $\PP{|\overline{Z} - \widetilde{Z}| > \Delta_{3n}} \leq D(q) \left(\gamma + \log n / n\right)$, where
\begin{align*}
    \Delta_{3n} = \frac{b_nL_n}{\gamma^{1/2} n^{1/2 - 1/q}} + \frac{(b_n)^{1/2}L_n^{3/4}}{\gamma^{1/2} n^{1/4}} + \frac{(b_nL_n^2)^{1/3}}{\gamma^{1/3} n^{1/6}},
\end{align*} $L_n = d(q) V_n(\log n \vee \log(A_n b_n))$ and $d(q), D(q)$ are constants that depend only on $q$.

Combining that with the previous bounds, and using the triangle inequality, we have $\PP{|Z - \widetilde{Z}| > c_0^{-1}\epsilon_n + \Delta_{1n} + \Delta_{2n} + \Delta_{3n}} \leq D(q)\left(\gamma + \log n / n\right) + c/\log n$.

Finally, Dudley's Theorem (see, e.g., Corollary 2.2.8 in \citet{vaart1997weak}) implies that, for our class of functions,
\[\mathrm{E}[\widetilde{Z}] \leq \chi \sqrt{V_n \log(A_n b_n)} \] for an absolute constant $\chi > 0$.

By \cref{lem:gausemp2},
\[\sup_{t}\left|\PP{Z \leq t} - \P{\widetilde{Z} \leq t}\right| \leq \kappa r_{1n} \left(\chi \sqrt{V_n \log(A_n b_n} + \sqrt{1 \vee \log(1 / r_{1n})}\right) + r_{2n},\]
for an absolute constant $\kappa > 0$, where $r_{1n} = c_0^{-1}\epsilon_n + \Delta_{1n} + \Delta_{2n} + \Delta_{3n}$ and $r_{2n} = D(q)\left(\gamma + (\log n) / n\right) + c/\log n$.

\section{Auxiliary Results}
\begin{lemma}[Maximal inequality, Lemma M.2 of \citealp{belloni2018uniformly}]\label{lem:max}
Let $\F$ be a set of suitably measurable functions $f: \mathcal{W} \rightarrow \RR$, equipped with a measurable envelope $F: \mathcal{W} \rightarrow \RR$, $F \geq \sup_{f\in\F}|f|$, such that $\Vert F \Vert_{P,q} \leq M < \infty$ for some $q \geq 2$. Let $\sigma^2$ be any positive constant such that $\sup_{f\in\F} \Vert f \Vert_{P,2}^2 \leq \sigma^2 \leq \Vert F \Vert_{P,2}^2$. Suppose that there exist constants $a \geq e$ and $v \geq 1$ such that:
\begin{align*}
   \log \sup_Q N(\varepsilon \Vert F \Vert_{Q,2}, \F, \Vert \cdot \Vert_{Q,2}) \leq v \log(a/\varepsilon), \qquad \text{for all } 0 < \varepsilon \leq 1.
\end{align*}

Then, with probability at least $1 - c/\log n$, \[\sup_{f \in \F}|\Gn f| \leq K(q,c) \left(\sigma \sqrt{v \log(a)} + vn^{1/q - 1/2}M\log(a) \right),\] where $K(q,c) > 0$ is a constant depending only on $q$ and $c$.
\end{lemma}

\begin{lemma}[Algebra for covering entropies, Lemma L.1 of \citealp{belloni2018uniformly}]\label{lem:entr}
For any measurable classes of functions $\F$ and $\F'$ mapping $\mathcal{W}$ to $\RR$,
\begin{multline*}
 \log N(\varepsilon \Vert F + F'\Vert_{Q,2}, \F + \F', \Vert \cdot \Vert_{Q,2}) \leq \\ \log N(\tfrac{\varepsilon}{2} \Vert F\Vert_{Q,2}, \F, \Vert \cdot \Vert_{Q,2})  + \log N(\tfrac{\varepsilon}{2} \Vert F'\Vert_{Q,2}, \F', \Vert \cdot \Vert_{Q,2}).
\end{multline*}
\end{lemma}

\begin{lemma}[High-dimensional Gaussian approximation, Corollary 2.1 of \citealp{chernozhukov2021nearly}]\label{lem:gaus}
Let $X_1, \ldots, X_n$ be a sequence of centered independent random vectors in $\RR^p$, and $W = n^{-1/2}\sum_{i=1}^n X_i$, where $\E{WW'} = \Sigma$ with unit diagonal entries. Let $\lambda_{\min} > 0$ denote the smallest eigenvalue of $\Sigma$. Consider the class $\mathcal{A}$ of hyperrectangles in $\RR^p$, i.e., sets of the form $\times_{j=1}^p (a_j, b_j]$ for some $-\infty \leq a_j \leq b_j \leq \infty$, $j = 1, \, \ldots, \, p$. Under the conditions below, we provide a bound on: \[\varrho = \sup_{A \in \mathcal{A}} |\P{W \in A} - \P{Z \in A}|, \qquad Z \sim \mathcal{N}(0, \Sigma).\]

Define three alternative conditions (from most to least restrictive):
\begin{enumerate}[(E.1)]
    \item $|X_{ij}| \leq b_n$ for all $i =1, \ldots, n$ and $j = 1, \, \ldots, \, p$ almost surely.
    \item $\Vert X_{ij} \Vert_{\psi_2} \leq b_n$ for all $i =1, \ldots, n$ and $j = 1, \, \ldots, \, p$, where $\Vert X \Vert_{\psi_2} = \inf \{C > 0 \, : \, \E{\exp(X^2/C^2)} \leq 2\}$.
    \item $\Vert \max_{1\leq j \leq p} |X_{ij}| \Vert_{P,q} \leq b_n$ for all $i =1, \ldots, n$ and some $q \geq 4$.
\end{enumerate}
Define also:
\begin{enumerate}[(M)]
    \item $n^{-1}\sum_{i=1}^n \E{X_{ij}^4} \leq b_n^2$ for all $j = 1, \, \ldots, \, p$.
\end{enumerate}
Then: \begin{enumerate}[(i)]
    \item Under (E.1), we have: \[\varrho \leq \frac{C b_n (\log p)^{3/2} \log n}{\sqrt{n} \lambda_{\min}}\] for an absolute constant $C > 0$.
    \item Under (E.2) and (M), we have: \[\varrho \leq C \left\lbrace \frac{b_n (\log p)^{3/2} \log n}{\sqrt{n} \lambda_{\min}} + \frac{b_n^2 (\log p)^{2}}{\sqrt{n\lambda_{\min}}} \right\rbrace\] for an absolute constant $C > 0$.
    \item Under (E.3) and (M), we have: \[\varrho \leq C \left\lbrace \frac{b_n (\log p)^{3/2} \log n}{\sqrt{n} \lambda_{\min}} + \frac{b_n^2 (\log p)^{2} \log n}{n^{1-2/q} \lambda_{\min}} + \left[\frac{b_n^q (\log d)^{3q/2-4} \log n \log (pn)}{n^{q/2 - 1} (\lambda_{\min})^{q/2}} \right]^\frac{1}{q-2}\right\rbrace\] for a constant $C(q) > 0$ depending only on $q$.
\end{enumerate}
\end{lemma}

\begin{lemma}[Anti-concentration inequality, Theorem 3 and Corollary 1 of \citealp{chernozhukov2014comparison}]\label{lem:antic}
Let $(X_1, \ldots, X_p)$ be a centered Gaussian random vector, with $\sigma = \mathrm{E}[X_j^2]$ for all $j = 1, \ldots, p$. Then, for all $\epsilon > 0$, \[\sup_{x \in \RR} \mathrm{P}(\max_{1\leq j \leq p} X_j \in [x - \epsilon, x + \epsilon]) \leq 12\epsilon \sqrt{\log p} / \sigma.\]
\end{lemma}

\begin{lemma}[Gaussian approximation to suprema of empirical processes, Corollary 2.2 of \citealp{chernozhukov2014gaussian}]\label{lem:gausemp1}
Let $\F$ be a set of suitably measurable functions $f: \mathcal{W} \rightarrow \RR$, equipped with a measurable envelope $F: \mathcal{W} \rightarrow \RR$, $F \geq \sup_{f\in\F}|f|$. Let $\sigma^2$ be any positive constant such that $\sup_{f\in\F} \Vert f \Vert_{P,2}^2 \leq \sigma^2 \leq \Vert F \Vert_{P,2}^2$.  Assume that, for some $ b \geq \sigma$ and $q \geq 4$, we have $\sup_{f \in \F} \E[p]{|f|^k} \leq \sigma^2 b^{k-2}$ for $k = 2, 3, 4$, and $\Vert F \Vert_{P,q} \leq b$. Suppose that there exist constants $a \geq e$ and $v \geq 1$ such that:
\begin{align*}
   \log \sup_Q N(\varepsilon \Vert F \Vert_{Q,2}, \F, \Vert \cdot \Vert_{Q,2}) \leq v \log(a/\varepsilon), \qquad \text{for all } 0 < \varepsilon \leq 1.
\end{align*}
Let $Z = \sup_{f \in \F} \Gn f$. Then, for every $\gamma \in (0,1)$ there exists a random variable $\widetilde{Z} \overset{d}{=} \sup_{f \in \F} G_P f$ such that \begin{align*}
    \PP{|Z - \widetilde{Z}| > \frac{bL_n}{\gamma^{1/2} n^{1/2 - 1/q}} + \frac{(b\sigma)^{1/2}L_n^{3/4}}{\gamma^{1/2} n^{1/4}} + \frac{(b\sigma^2L_n^2)^{1/3}}{\gamma^{1/3} n^{1/6}}} \leq C\left(\gamma + \frac{\log n}{n}\right),
\end{align*} where $L_n = Bv(\log n \vee \log(ab/\sigma))$ and $B, C$ are constants that depend only on $q$. Here, $G_P$ denotes a tight mean-zero Gaussian process with covariance function $\E{G_P(f)G_P(g)} = \E{f(W)g(W)}$ for all $f, g \in \F$.
\end{lemma}

\begin{lemma}[Gaussian approximation to suprema of empirical processes, Kolmogorov distance version, Lemma 2.3 of \citealp{chernozhukov2014gaussian}]\label{lem:gausemp2}
Under the same conditions as the previous lemma, suppose that there exist constants such that $c_0 \leq \Vert f \Vert_{P,2} \leq C_0$ for all $f \in \F$. For a random variable $Z$, suppose that there exist constants $r_1, r_2 > 0$ and $\widetilde{Z} \overset{d}{=} \sup_{f \in \F} G_P f$ such that $\PP{|Z - \widetilde{Z}| > r_1} \leq r_2$. Then, \[\sup_{t}\left|\PP{Z \leq t} - \P{\widetilde{Z} \leq t}\right| \leq \kappa r_1 \left(\mathrm{E}[\widetilde{Z}] + \sqrt{1 \vee \log(c_0 / r_1)}\right) + r_2,\] where $\kappa$ is a constant depending only on $c_0$ and $C_0$.
\end{lemma}
\end{appendices}
\end{document}